\documentclass[american,aps,pra,reprint, superscriptaddress]{revtex4-1}
\usepackage{amsmath,amscd,amsthm,amssymb}
\usepackage{mathrsfs}
\usepackage{graphicx}
\usepackage{epsf,epstopdf}
\usepackage{caption3}
\usepackage[all]{xy}
\usepackage{tikz}
\usepackage{enumerate}
\usepackage{mathtools}
\usepackage[unicode=true,pdfusetitle, bookmarks=true,bookmarksnumbered=false,bookmarksopen=false, breaklinks=false,pdfborder={0 0 0},backref=false,colorlinks=false] {hyperref}
\hypersetup{colorlinks,linkcolor=myurlcolor,citecolor=myurlcolor,urlcolor=myurlcolor}
\usepackage{graphics, times, braket, color, bm}
\usepackage[up]{subfigure}
\usepackage{cleveref}
\usepackage{bbm}

\definecolor{myurlcolor}{rgb}{0,0,0.7}

\def\be{\begin{equation}}
\def\ee{\end{equation}}
\def\bea{\begin{eqnarray*}}
\def\eea{\end{eqnarray*}}

\theoremstyle{plain}
\newtheorem{thrm}{\protect\theoremname}

\newtheorem{exam}[thrm]{Example}
\newtheorem{cor}[thrm]{Corollary}
\providecommand{\theoremname}{Theorem}

\newcommand{\iinner}[2]{\langle #1 | #2\rangle}

\DeclareMathOperator{\trace}{Tr}
\newcommand{\ptr}[2]{\trace_{#1}({#2})}
\newcommand{\tr}[1]{\ptr{}{#1}}

\newcommand*{\myproofname}{Proof}

\makeatother
\newcommand{\DOI}[2]{\href{https://doi.org/#1}{#2}}

\renewcommand{\Re}{\operatorname{Re}}

\theoremstyle{definition}

\theoremstyle{remark}

\graphicspath{{figs/}}

\begin{document}

\title{Quantum-State Texture as a Resource: Measures and Nonclassical Interdependencies}

\author{Yuntao Cui}
 \email{YunTao\_Cui@outlook.com}
 \affiliation{School of Mathematical Sciences, Harbin Engineering University, Harbin 150001, People's Republic of China}

 \author{Zhaobing Fan}
 \email{Corresponding author: fanzhaobing@hrbeu.edu.cn}
 \affiliation{School of Mathematical Sciences, Harbin Engineering University, Harbin 150001, People's Republic of China}

\author{Sunho Kim}
 \email{Corresponding author: kimsunho81@hrbeu.edu.cn}
 \affiliation{School of Mathematical Sciences, Harbin Engineering University, Harbin 150001, People's Republic of China}

\begin{abstract}
Quantum-state texture is a newly recognized quantum resource that has garnered attention with the advancement of quantum theory. In this work, we address several key aspects of quantum-state texture resource theory, including the quantification of quantum-state texture, quantum state transformation under free operations, and the relationships between quantum resources. We first propose two new measures of quantum-state texture and introduce a specific functional form for constructing such measures via the convex roof method. Then, we determine the maximum probability of quantum state transformation under free operations. Finally, we establish connections between quantum-state texture and other prominent quantum resources, such as coherence, imaginarity, and predictability. Our research contributes to the measure theory of quantum-state texture and enriches the overall framework of quantum-state texture resource theory.
\end{abstract}
\maketitle

\section{Introduction}

The notion of quantum-state texture (QST) was first proposed by Parisio \cite{ref1}. QST can be understood intuitively as follows: Given a computational basis ${\ket{i}}$, we can visualize the density matrix $\rho$ as a surface plot, where each matrix entry $\rho_{ij}$ represents an altitude. We treat the row and column indices as the first two dimensions. The real part of each $\rho_{ij}$ defines the altitude in one surface plot, and the imaginary part defines the altitude in an analogous plot. Thus, each quantum state corresponds to a pair of three-dimensional surface plots. From this perspective, these plots generally exhibit unevenness, or texture. The simplest possible plot is one where all altitudes are filled with the same number. The only quantum state that produces such a uniform plot is the state
\begin{equation}
f=\ket{f}\bra{f}
\end{equation}
where $\ket{f}=\frac{1}{\sqrt{d}}\sum_{i=0}^{d-1}\ket{i}$ and $d$ is the dimension of the Hilbert space $\mathcal{H}$. This state $f$ is a non-texture state and constitutes the zero-resource set of the theory. The associated free operations, corresponding to completely positive and trace-preserving (CPTP) maps $\Lambda$, are implemented by Kraus operators: $\Lambda(\rho)=\sum_{n}K_{n}\rho K_{n}^{\dagger}$, with $\sum_{n}K_{n}^{\dagger}K_{n}=\mathbbm{1}$. These free maps must satisfy $\Lambda(f)=f$, which implies $K_{n}\ket{f}\propto\ket{f}$ for all $n$.

Parisio \cite{ref1} demonstrated that a circuit layer containing at least one CNOT gate can be fully characterized by using randomized input states and analyzing the texture of the output qubits. Notably, this process does not require tomographic protocols or ancillary systems. As a quantum resource, QST is analogous to quantum coherence, which is a key component in emerging quantum technologies such as quantum metrology \cite{ref2,ref3}, quantum computing \cite{ref4}, nanoscale thermodynamics \cite{ref5,ref6,ref7,ref8}, and biological systems \cite{ref9,ref10,ref11}. Therefore, QST may also have a positive impact on quantum information tasks.

As a newly emerging quantum resource, quantum-state texture (QST) has spurred the development of measurement techniques. While Ref.~\cite{ref12} proposed various measurement methods, this work complements the framework by introducing two distance-based QST measures. Unlike Ref.~\cite{ref12}, we employ affinity and Tsallis relative $\alpha$-entropy for this purpose. We further demonstrate that the QST measure defined by the Hellinger distance is a special case of the measure defined by affinity. Additionally, we propose a general method for constructing QST measures, denoted as $\Upsilon_{g}$, using monotonic functions $g$, and provide a specific example that plays a pivotal role in subsequent sections.

Inspired by the resource theories of entanglement~\cite{ref13,ref14}, coherence~\cite{ref15,ref16,ref17,ref18,ref19}, and imaginarity~\cite{ref20,ref21}, numerous resource-related tasks remain to be explored. This paper focuses on two fundamental tasks within the quantum resource theory of QST. The first is state transformation via free operations, a central question in any resource theory. We address the stochastic transformation of states under QST free operations and derive the optimal transformation probability, $p$. The second task involves elucidating the relationships between different quantum resources. We investigate the connections between QST and other resources---primarily coherence, imaginarity, and predictability---and establish several relations in the form of equations and inequalities under different QST measures.

The remainder of this paper is organized as follows. In Sec.~\uppercase\expandafter{\romannumeral2}, we introduce three distinct QST measures and provide a specific example. Section~\uppercase\expandafter{\romannumeral3} is devoted to state transformation via free operations. In Sec.~\uppercase\expandafter{\romannumeral4}, we discuss the relationships between QST and other quantum resources. Finally, we present our conclusions in Sec.~\uppercase\expandafter{\romannumeral5}.

\section{Quantum-state texture measures}\label{sec:2}

The basic requirements of a QST measure were introduced in Ref. \cite{ref1}. Specifically, for arbitrary QST measure $\Upsilon$, it must satisfy the following conditions:\\
$(1).$ $\Upsilon(\rho)\geq0$ and $\Upsilon(f)=0$;\\
$(2).$ $\Upsilon(\rho)\geq\Upsilon[\Lambda(\rho)]$, where $\Lambda$ is completely positive and trace-preserving maps satisfied $\Lambda(f)=f$, are affected by Kraus operators, $\Lambda(\rho)=\sum_{n}K_{n}\rho K_{n}^{\dagger}$, for which $\sum_{n}K_{n}^{\dagger}K_{n}=\mathbbm{1}$;\\
$(3).$ $\Upsilon$ is convex, i.e., $\Upsilon(\sum_{i}p_{i}\rho_{i})\leq\sum_{i}p_{i}\Upsilon(\rho_{i})$, $\sum_{i}p_{i}=1$, $p_{i}\geq0$.

QST, as a quantum resource, may also hold potenital value due to its close connection with coherence. Thus, exploring possible candidates for QST measurement and analyzing them is an interesting and relevant problem. The author of Ref. \cite{ref1} defined a state texture quantifier
\begin{equation}
	\mathcal{R}(\rho)=-\ln\bra{f}\rho\ket{f},
\end{equation}
which they refered to it as "state rugosity". And a variety of QST meaurements based on distance definition have been proposed, such as trace distance, geometric distance, Bure distance and fidelity \cite{ref12}. In this chapter, we also propose several measurements for QST, and try to establish some connections with the proposed measurements.

The minimum distance between quantum states $\rho$ and non-resource states $\sigma\in\mathbbm{I}$ is widely used as the most fundamental method of resource quantification in various quantum resource theories. In QST, the distance between quantum state $\rho$ and the non-texture state $f$ is used, as the free state is solely $f$ \cite{ref1,ref12}.
\begin{equation}
	\Upsilon(\rho)=D(\rho,f).
\end{equation}
Thus, we can define QST measure base on different distances.

\subsection{Quantum-state texture measurement based on affinity}

The affinity of two quantum states is defined as
\begin{equation}
	A(\rho,\sigma)=\tr{\sqrt{\rho}\sqrt{\sigma}}.
\end{equation}
Quantum affinity, similar to fidelity \cite{ref22}, describes how close two quantum states are. It also possess all the properties of fidelity. This quantity is more useful in quantum detection and estimation theory. The notion of affinity has been extended to $\alpha$-affinity $(0<\alpha<1)$, which is defined as
\begin{equation}
	A_{\alpha}(\rho,\sigma)=\tr{\rho^{\alpha}\sigma^{1-\alpha}}
\end{equation}
with $\alpha\in(0,1)$. The $\alpha$-affinity satisfies the following properties \cite{ref23}:\\
(1). $\alpha$-affinity is bounded:
\begin{equation}
	0\leq A_{\alpha}(\rho,\sigma)\leq1
\end{equation}
and $A_{\alpha}(\rho,\sigma)=1$ if and only if $\rho=\sigma$ for all values of $\alpha$.\\
(2). Monotonic:
\begin{equation}
	A_{\alpha}(\Phi(\rho),\Phi(\sigma))\geq A_{\alpha}(\rho,\sigma)
\end{equation}
under a completely positive and trace preserving (CPTP) map.\\
(3). Joint concavity:
\begin{equation}
	A_{\alpha}(\sum_{i}P_{i}\rho_{i},\sum_{i}P_{i}\sigma_{i})\geq\sum_{i}P_{i}A_{\alpha}(\rho_{i},\sigma_{i}).
\end{equation}
So, we can define a measure of QST base on $\alpha$-affinity
\begin{equation}
	\Upsilon_{A}(\rho)=1-A_{\alpha}(\rho,f)
\end{equation}
where $0<\alpha<1$ and $f=\ket{f}\bra{f}$ is non-texture state. It's easy to prove that.

\begin{thrm}
	$\Upsilon_{A}(\rho)=1-A_{\alpha}(\rho,f)$ is a QST measure.
\end{thrm}
\begin{proof}
	$\Upsilon_{A}(\rho)$ satisfies the propertise of QST measure.\\
	(1). Non-negative:
	\begin{equation}
		0\leq \Upsilon_{A}(\rho)=1-A_{\alpha}(\rho,f)\leq1
	\end{equation}
	and $\Upsilon_{A}(\rho)=0$ if and only if $\rho=f$.\\
	(2). Monotonic:
	\begin{equation}
		\begin{aligned}
			\Upsilon_{A}(\Lambda(\rho)) &=1-A_{\alpha}(\Lambda(\rho),f)\\ &=1-A_{\alpha}(\Lambda(\rho),\Lambda(f))\\ &\leq1-A_{\alpha}(\rho,f)\\ &=\Upsilon_{A}(\rho)
		\end{aligned}
	\end{equation}
	where $\Lambda$ is a completely positive and trace-preserving maps satisfied $\Lambda(f)=f$, that is free operator.\\
	(3). Convex:
	\begin{equation}
		\begin{aligned}
			\Upsilon_{A}(\sum_{i}P_{i}\rho_{i}) &=1-A_{\alpha}(\sum_{i}P_{i}\rho_{i},f)\\ &=1-A_{\alpha}(\sum_{i}P_{i}\rho_{i},\sum_{i}P_{i}f)\\ &\leq 1-\sum_{i}P_{i}A_{\alpha}(\rho_{i},f)\\ &=\sum_{i}P_{i}(1-A_{\alpha}(\rho_{i},f))\\ &=\sum_{i}P_{i}\Upsilon_{A}(\rho_{i})
		\end{aligned}
	\end{equation}
	the joint concavity of $A_{\alpha}$ is applied here.
\end{proof}

The Hellinger distance of quantum states $\rho$ and $\sigma$ is defined as \cite{ref24}:
\begin{equation}
	D_{H}(\rho,\sigma)=\tr{\sqrt{\rho}-\sqrt{\sigma}}^{2}.
\end{equation}
Referring to the distance measure of resource theory, the distance between quantum state $\rho$ and non-texture state $f$ can be used to measure the QST of $\rho$. So the Hillinger distance between $\rho$ and $f$
\begin{equation}
	D_{H}(\rho,f)=\tr{\sqrt{\rho}-\sqrt{f}}^{2}
\end{equation}
should be a QST measure. In fact, we can see that $D_{H}(\rho,f)$ and $\Upsilon_{A}(\rho)$ are related.
\begin{equation}
	\begin{aligned}
		D_{H}(\rho,f) &=\tr{\sqrt{\rho}-\sqrt{f}}^{2}\\ &=\tr{\rho+f-2\sqrt{\rho}\sqrt{f}}\\ &=2[1-\tr{\sqrt{\rho}\sqrt{f}}]\\ &=2[1-A_{\frac{1}{2}}(\rho,f)].
	\end{aligned}
\end{equation}
In this way, $D_{H}(\rho,f)$ is a special form of $\Upsilon_{A}(\rho)$ in the case of $\alpha=\frac{1}{2}$. So, $D_{H}(\rho,f)$ also satisifes the properties of QST measure, the only different point is
\begin{equation}
	0\leq D_{H}(\rho,f)\leq2.
\end{equation}

\subsection{Quantum-state texture measurement based on Tsallis relative $\alpha$ entropy}

For $0\leq\alpha\neq1$, the Tsallis relative $\alpha$ entropy is defined as \cite{ref25,ref26}:
\begin{equation}
	D_{\alpha}(p\|q)=\frac{1}{\alpha-1}(\sum_{j}p_{j}^{\alpha}q_{j}^{1-\alpha}-1).
\end{equation}
The quantum relative entropy is expressed as \cite{ref27}:
\begin{equation}
	\begin{aligned}
		D_{1}(\rho\|\sigma)=
		\begin{cases}
			\tr{\rho\ln\rho-\rho\ln\sigma} & \mathbf{ran}(\rho)\subset \mathbf{ran}(\sigma)\\ +\infty & otherwise
		\end{cases},
	\end{aligned}
\end{equation}
where $\mathbf{ran}(\rho)$ is the range of $\rho$.
For $\alpha\in(1,+\infty)$, the Tsallis $\alpha$ divergence is defined as \cite{ref27,ref28}:
\begin{equation}
	\begin{aligned}
		D_{\alpha}(\rho\|\sigma)=
		\begin{cases}
			\frac{\tr{\rho^{\alpha}\sigma^{1-\alpha}}-1}{\alpha-1} & \mathbf{ran}(\rho)\subset \mathbf{ran}(\sigma)\\ +\infty & otherwise
		\end{cases}.
	\end{aligned}
\end{equation}
And $D_{\alpha}(\rho\|\sigma)$ has the following properties \cite{ref29}:

1. $D_{\alpha}(\rho\|\sigma)\geq0$, with equality if and only if $\rho=\sigma$.

2. For $\alpha\in(0,2]$, $D_{\alpha}(\Phi(\rho)\|\Phi(\sigma))\leq D_{\alpha}(\rho\|\sigma)$.

3. For $\alpha\in(0,2]$, $D_{\alpha}(\sum_{n}p_{n}\rho_{n}\|\sum_{n}p_{n}\sigma_{n})\leq\sum_{n}p_{n}D_{\alpha}(\rho_{n}\|\sigma_{n})$.

\begin{thrm}
	For $\alpha\in(0,2]$, $D_{\alpha}(\rho||f)$ is a QST measure.
\end{thrm}
\begin{proof}
	All the requirements can be derived from the properties of $\alpha$ relative entropy easily.
	
	1. $D_{\alpha}(\rho\|f)\geq0$, $D_{\alpha}(f\|f)=0$. And $D_{\alpha}(f^{\perp}\|f)=1$ reach the maximum value, futhermore the linear combination of $\ket{f^{\perp}}$ also reaches the maximum value.
	
	2. According to condition 2, available
	\begin{equation}
		D_{\alpha}(\Lambda(\rho)\|f)=D_{\alpha}(\Lambda(\rho)\|\Lambda(f))\leq D_{\alpha}(\rho\|f).
	\end{equation}
	
	3. According to condition 3, available
	\begin{equation}
		\begin{aligned}
			D_{\alpha}(\sum_{n}P_{n}\rho_{n}\|f) &=D_{\alpha}(\sum_{n}P_{n}\rho_{n}\|\sum_{n}P_{n}f)\\ &\leq\sum_{n}P_{n}D_{\alpha}(\rho_{n}\|f).
		\end{aligned}
	\end{equation}
\end{proof}

\subsection{Quantum-state texture measurement based on function}

Similar to other resources \cite{ref30,ref31,ref32,ref33}, we can also measure QST through functions. In the first, let us define a function $g:[0,1]\to[0,1]$ satisfying the following conditions:

1. $g(1)=0$;

2. $g$ is monotonically decreasing;

3. $g$ is concave: $g(\lambda x+(1-\lambda)y)\geq\lambda g(x)+(1-\lambda)g(y)$ for all $\lambda\in[0,1]$ and all $x,y\in[0,1]$.

For any pure state $\ket{\psi}\bra{\psi}$, define
\begin{equation}\label{22}
	\Upsilon_{g}(\ket{\psi})=g(|\iinner{f}{\psi}|^{2}),
\end{equation}
and for any mixed state $\rho$, define
\begin{equation}\label{23}
	\Upsilon_{g}(\rho)=\min_{\{p_{i},\ket{\psi_{i}}\}}\{\sum_{i}p_{i}\Upsilon_{g}(\ket{\psi_{i}})\}
\end{equation}
where $\rho=\sum_{i}p_{i}\ket{\psi_{i}}\bra{\psi_{i}}$.

\begin{thrm}
	For every function $g$ with conditions 1 to 3, the $\Upsilon_{g}$ defined in \eqref{22} and \eqref{23} is a QST measure.
\end{thrm}
\begin{proof}
	
	1. Non-negative:
	\begin{equation}
		\Upsilon_{g}(\rho)=g(|\iinner{f}{\psi}\iinner{\psi}{f}|^{2})=g(\frac{\sum_{i,j}\psi_{i}\psi_{j}^{\ast}}{d})
	\end{equation}
	$\sum_{i,j}\psi_{i}\psi_{j}^{\ast}=d$ reaches its maximum value when $\ket{\psi}=\ket{f}$, that is $\psi_{i}=\frac{1}{\sqrt{d}}$, meanwhile $\Upsilon_{g}(f)=g(1)=0$ reaches minimum value. So $\Upsilon_{g}(\rho)\geq0$ because of $g$ monotonically decreasing.
	
	2. Monotonicity:
	For pure states, there is
	\begin{equation}
		\begin{aligned}
			&\sum_{n}|\bra{f}K_{n}\ket{\psi}|^{2}\\ &=\sum_{n}\bra{f}K_{n}\ket{\psi}\bra{\psi}K_{n}^{\dagger}\ket{f}\\ &=\sum_{n}\bra{f}K_{n}(\alpha\ket{f}+\beta K_{n}\ket{f^{\perp}})(\alpha^{\ast}\bra{f}+\beta^{\ast}\bra{f^{\perp}})K_{n}^{\dagger}\ket{f}\\ &=\sum_{n}(\alpha a_{n}+\beta\bra{f}K_{n}\ket{f^{\perp}})(\alpha^{\ast}a_{n}^{\ast}+\beta^{\ast}\bra{f^{\perp}}K_{n}^{\dagger}\ket{f})\\ &\geq|\alpha|^{2}\\ &=|\iinner{f}{\psi}|^{2}
		\end{aligned}
	\end{equation}
	so it follows that
	\begin{equation}\label{26}
		\begin{aligned}
			\Upsilon_{g}(\ket{\psi}) &=g(|\iinner{f}{\psi}|^{2})\\
			&\geq g(\sum_{n}|\bra{f}K_{n}\ket{\psi}\bra{\psi}K_{n}^{\dagger}\ket{f}|)\\
			&\geq\sum_{n}\tr{K_{n}\ket{\psi}\bra{\psi}K_{n}^{\dagger}}g(\frac{|\bra{f}K_{n}\ket{\psi}\bra{\psi}K_{n}^{\dagger}\ket{f}|}{\tr{K_{n}\ket{\psi}\bra{\psi}K_{n}^{\dagger}}})\\
			&=\sum_{n}\tr{K_{n}\ket{\psi}\bra{\psi}K_{n}^{\dagger}}\Upsilon_{g}(\frac{K_{n}\ket{\psi}\bra{\psi}K_{n}^{\dagger}}{\tr{K_{n}\ket{\psi}\bra{\psi}K_{n}^{\dagger}}})\quad \ \
		\end{aligned}
	\end{equation}
	because of $g$ monotonically decresing and concave.
	
	Then for mixed states $\rho$, let $\rho=\sum_{i}q_{i}\ket{\psi_{i}}\bra{\psi_{i}}$ is an optimal pure-state ensembel with $\Upsilon_{g}(\rho)=\sum_{i}q_{i}\Upsilon_{g}(\ket{\psi_{i}}\bra{\psi_{i}})$. Then
\begin{widetext}
    \begin{equation}
		\begin{aligned}
			\Upsilon_{g}(\Lambda(\rho))
			&=\Upsilon_{g}(\sum_{n}K_{n}\rho K_{n}^{\dagger})\\
			&\leq\sum_{n}\tr{K_{n}\rho K_{n}^{\dagger}}\Upsilon_{g}(\frac{K_{n}\rho K_{n}^{\dagger}}{\tr{K_{n}\rho K_{n}^{\dagger}}})\\
			&=\sum_{n}\tr{K_{n}\rho K_{n}^{\dagger}}\Upsilon_{g}(\frac{K_{n}\sum_{i}q_{i}\ket{\psi_{i}}\bra{\psi_{i}}K_{n}^{\dagger}}{\tr{K_{n}\rho K_{n}^{\dagger}}})\\
			&=\sum_{n}\tr{K_{n}\rho K_{n}^{\dagger}}\Upsilon_{g}(\frac{\sum_{i}q_{i}K_{n}\ket{\psi_{i}}\bra{\psi_{i}}K_{n}^{\dagger}}{\tr{K_{n}\rho K_{n}^{\dagger}}})\\
			&\leq\sum_{n}\tr{K_{n}\rho K_{n}^{\dagger}}\sum_{i}\frac{q_{i}}{\tr{K_{n}\rho K_{n}^{\dagger}}}\tr{K_{n}\ket{\psi_{i}}\bra{\psi_{i}}K_{n}^{\dagger}}\Upsilon_{g}(\frac{K_{n}\ket{\psi_{i}}\bra{\psi_{i}}K_{n}^{\dagger}}{\tr{K_{n}\ket{\psi_{i}}\bra{\psi_{i}}K_{n}^{\dagger}}})\\
			&=\sum_{i}q_{i}(\sum_{n}\tr{K_{n}\ket{\psi_{i}}\bra{\psi_{i}}K_{n}^{\dagger}}\Upsilon_{g}(\frac{K_{n}\ket{\psi_{i}}\bra{\psi_{i}}K_{n}^{\dagger}}{\tr{K_{n}\ket{\psi_{i}}\bra{\psi_{i}}K_{n}^{\dagger}}}))\\
			&\leq\sum_{i}q_{i}\Upsilon_{g}(\ket{\psi_{i}}\bra{\psi_{i}})\\
			&=\Upsilon_{g}(\rho)
		\end{aligned}
	\end{equation}
\end{widetext}
    and the convexity of $\Upsilon_{g}$ and Eq.\eqref{26} are used here.
	
	3. Convexity:
	Let $\rho=\sum_{i}p_{i}\rho_{i}$, and $\rho_{i}=\sum_{j}q_{ij}\rho_{ij}$ be the optimal pure-state ensemble of each $\rho_{i}$, which satisfies $\Upsilon_{g}(\rho_{i})=\sum_{j}q_{ij}\Upsilon_{g}(\rho_{ij})$. Then
	\begin{equation}
		\begin{aligned}
			\Upsilon_{g}(\rho) &=\Upsilon_{g}(\sum_{i}p_{i}\sum_{j}q_{ij}\rho_{ij})\\ &=\Upsilon_{g}(\sum_{i,j}p_{i}q_{ij}\rho_{ij})\\ &=\min_{p_{k},\ket{\psi_{k}}}\sum_{k}p_{k}\Upsilon_{g}(\ket{\psi_{k}})\\ &\leq\sum_{i,j}p_{i}q_{ij}\Upsilon_{g}(\rho_{ij})\\ &=\sum_{i}p_{i}\Upsilon_{g}(\rho_{i})
		\end{aligned}
	\end{equation}
	and the inequality follows from the definition of $\Upsilon_{g}(\rho)$.
\end{proof}
Using Theorem 3, we can construct many interesting QST measures. We provide one example here that plays an important role in the state transformations discussed later.

\begin{exam}
	Let g(x)=$1-|x|^{2}$, then
	\bea
		&&\Upsilon_{g}(\ket{\psi})=1-|\iinner{f}{\psi}|^{2};\\
		&&\Upsilon_{g}(\rho)=\min_{\{p_{i},\ket{\psi_{i}}\}}\sum_{i}\Upsilon_{g}(\ket{\psi_{i}}).
	\eea
\end{exam}
This measure corresponds to the geometric measure presented in Ref. \cite{ref12}. However, we demonstrate that for mixed states, the geometric measure of QST is independent of the convex-roof construction, and we present its generalized form:
\begin{equation}
	\begin{aligned}
		\Upsilon_{g}(\rho) &=\min_{\{p_{i},\ket{\psi_{i}}\}}\sum_{i}\Upsilon_{g}(\ket{\psi_{i}})\\ &=\min_{\{p_{i},\ket{\psi_{i}}\}}\sum_{i}p_{i}(1-|\bra{f}\ket{\psi_{i}}|^{2})\\ &=1-\max_{\{p_{i},\ket{\psi_{i}}\}}\sum_{i}p_{i}\bra{f}\ket{\psi_{i}}\bra{\psi_{i}}\ket{f}\\ &=1-\max_{\{p_{i},\ket{\psi_{i}}\}}\bra{f}(\sum_{i}p_{i}\ket{\psi_{i}}\bra{\psi_{i}})\ket{f}\\ &=1-\bra{f}\rho\ket{f}.
	\end{aligned}
\end{equation}

\section{State transformations via free operations}

One of the main questions in quantum resource theory is whether, for two given states $\rho$ and $\sigma$, there exists a free operation $\Lambda_{f}$ that transforms $\rho$ into $\sigma$:
\begin{equation}\label{30}
\sigma=\Lambda_{f}(\rho).
\end{equation}
The existence of such a transformation immediately implies that $\rho$ is more resourceful than $\sigma$, and in particular,
\begin{equation}
R(\rho)\geq R(\sigma)
\end{equation}
for any resource measure $R$.

Even if $\rho$ cannot be converted into $\sigma$ deterministically via a free operation (for instance, if $R(\rho) < R(\sigma)$), it might still be possible to achieve the conversion probabilistically, provided the resource theory allows for stochastic free operations. These are described by free Kraus operators ${K_{j}}$ such that $\sum_{j}K_{j}^{\dagger}K_{j} \leq \mathbbm{1}$. It is further reasonable to assume that any incomplete set of free operators ${K_{j}}$ can be completed with additional free Kraus operators ${L_{i}}$ such that
\begin{equation}
\sum_{i}L_{i}^{\dagger}L_{i} + \sum_{j}K_{j}^{\dagger}K_{j} = \mathbbm{1}.
\end{equation}
The maximal probability for converting $\rho$ into $\sigma$ is then defined as:
\begin{equation}
P(\rho\to\sigma) = \max_{{K_{j}}} \left\{\sum_{j} p_{j} : \sigma = \frac{ \sum_{j} K_{j} \rho K_{j}^{\dagger} }{ \sum_{j} p_{j} } \right\}
\end{equation}
with probabilities $p_{j} = \tr( K_{j} \rho K_{j}^{\dagger} )$, and the maximum is taken over all sets of free Kraus operators ${K_{j}}$. The existence of a deterministic free operation between $\rho$ and $\sigma$ as in Eq.~\eqref{30} is then equivalent to $P(\rho\to\sigma) = 1$.

Stochastic transformations of quantum resources such as entanglement, coherence, and imaginarity have been widely discussed \cite{ref34,ref35,ref36}. In the following, we will discuss stochastic transformations within the framework of QST.

\subsection{Stochastic transformations for pure states}

We now provide the maximal probability for converting a pure state $\ket{\psi}$ into another pure state $\ket{\phi}$ via free operations of QST.
\begin{thrm}
	The maximum probability for pure qubit state transformation $\ket{\psi}\to\ket{\phi}$ via free operations of QST is given by
	\begin{equation}
		\begin{aligned}			P(\ket{\psi}\to\ket{\phi})=\min\{\frac{\Upsilon_{g}(\ket{\psi})}{\Upsilon_{g}(\ket{\phi})},1\}=\min\{\frac{1-|\iinner{f}{\psi}|^{2}}{1-|\iinner{f}{\phi}|^{2}},1\}.\quad \ \
		\end{aligned}
	\end{equation}
\end{thrm}
\begin{proof}
	For pure states
	\begin{equation}
		\begin{aligned}
			&\ket{\psi}=\iinner{f}{\psi}\ket{f}+\sqrt{1-|\iinner{f}{\psi}|^{2}}\ket{f^{\perp}}=\alpha\ket{f}+\beta\ket{f^{\perp}};\\ &\ket{\phi}=\iinner{f}{\phi}\ket{f}+\sqrt{1-|\iinner{f}{\phi}|^{2}}\ket{f^{\perp}}=\mu\ket{f}+\nu\ket{f^{\perp}}.\quad \ \
		\end{aligned}
	\end{equation}
	We can obtain the matrix form
	\begin{equation}
		\begin{aligned}
			&\ket{\psi}\bra{\psi}=
			\begin{pmatrix}
				|\alpha|^{2} & \alpha\beta^{\ast} \\ \alpha^{\ast}\beta & |\beta|^{2}
			\end{pmatrix}=
			\begin{pmatrix}
				|\alpha|^{2} & |\alpha||\beta|e^{i(\alpha-\beta)} \\ |\alpha||\beta|e^{i(\beta-\alpha)} & |\beta|^{2}
			\end{pmatrix};\\
			&\ket{\phi}\bra{\phi}=
			\begin{pmatrix}
				|\mu|^{2} & \mu\nu^{\ast} \\ \mu^{\ast}\nu & |\nu|^{2}
			\end{pmatrix}=
			\begin{pmatrix}
				|\mu|^{2} & |\mu||\nu|e^{i(\mu-\nu)} \\ |\mu||\nu|e^{i(\nu-\mu)} & |\nu|^{2}
			\end{pmatrix}\quad \ \
		\end{aligned}
	\end{equation}
	where $\alpha=|\alpha|e^{i\alpha}$, $\beta=|\beta|e^{i\beta}$, $\mu=|\mu|e^{i\mu}$, $\nu=|\nu|e^{i\nu}$.
	
	If $\ket{\psi}$ has stronger texture than $\ket{\phi}$, the transformation $\ket{\psi}\to\ket{\phi}$ can be achieved unit probability \cite{ref1}. Then we consider the case $\ket{\psi}$ has weaker texture than $\ket{\phi}$, that is
	\begin{equation}
		|\iinner{f}{\psi}|\geq|\iinner{f}{\phi}|\ ,\ |\alpha|\geq|\mu|.
	\end{equation}
	We prove that convert probability $p$ msut less than $\frac{1-|\iinner{f}{\psi}|^{2}}{1-|\iinner{f}{\phi}|^{2}}$ first. Let
	\begin{equation}
		K_{0}=
		\begin{pmatrix}
			a & b \\ c & d
		\end{pmatrix}
	\end{equation}
	be a free operation of QST, and $K_{1}=\sqrt{\mathbbm{1}-K_{0}^{\dagger}K_{0}}$, here the basis of its matrix form is $\{\ket{f},\ket{f^{\perp}}\}$. It must satisfies $K_{0}\ket{f}=a\ket{f}$, which means
	\begin{equation}
		\begin{aligned}
			K_{0}\ket{f}\bra{f}K_{0}^{\dagger} &=
			\begin{pmatrix}
				a & b \\ c & d
			\end{pmatrix}
			\begin{pmatrix}
				1 & 0 \\ 0 & 0
			\end{pmatrix}
			\begin{pmatrix}
				a^{\ast} & c^{\ast} \\ b^{\ast} & d^{\ast}
			\end{pmatrix}\\ &=
			\begin{pmatrix}
				|a|^{2} & ac^{\ast} \\ a^{\ast}c & |c|^{2}
			\end{pmatrix}\\ &=|a|^{2}\ket{f}\bra{f}
		\end{aligned}
	\end{equation}
	so there must be $c=c^{\ast}=0$. Therefore, we can assume
	\begin{equation}
		K_{0}=
		\begin{pmatrix}
			a & b \\ 0 & c
		\end{pmatrix}
	\end{equation}
	and we are hoping for
	\begin{equation}
		\begin{aligned}
			K_{0}\ket{\psi}\bra{\psi}K_{0}^{\dagger}=q\ket{\phi}\bra{\phi}
		\end{aligned}
	\end{equation}
	which means
	\begin{widetext}
		\begin{equation}
			\begin{aligned}
				K_{0}\ket{\psi}\bra{\psi}K_{0}^{\dagger} &=
				\begin{pmatrix}
					a & b \\ 0 & c
				\end{pmatrix}
				\begin{pmatrix}
					|\alpha|^{2} & |\alpha||\beta|e^{i(\alpha-\beta)}\\ |\alpha||\beta|e^{i(\beta-\alpha)} & |\beta|^{2}
				\end{pmatrix}
				\begin{pmatrix}
					a^{\ast} & 0 \\ b^{\ast} & c^{\ast}
				\end{pmatrix}\\ &=
				\begin{pmatrix}
					|a|^{2}|\alpha|^{2}+a^{\ast}b|\alpha||\beta|e^{i(\beta-\alpha)}+ab^{\ast}|\alpha||\beta|e^{i(\alpha-\beta)}+|b|^{2}|\beta|^{2} & ac^{\ast}|\alpha||\beta|e^{i(\alpha-\beta)}+bc^{\ast}|\beta|^{2} \\ a^{\ast}c|\alpha||\beta|e^{i(\beta-\alpha)}+b^{\ast}c|\beta|^{2} & |c|^{2}|\beta|^{2}
				\end{pmatrix}\\ &=
				\begin{pmatrix}
					q|\mu|^{2} & q|\mu||\nu|e^{i(\mu-\nu)} \\ q|\mu||\nu|e^{i(\nu-\mu)} & q|\nu|^{2}
				\end{pmatrix}\\ &=q\ket{\phi}\bra{\phi}.
			\end{aligned}
		\end{equation}
	\end{widetext}
	Observing the element in the second row and second column of the matrix leads to the conclusion that
	\begin{equation}
		q|\nu|^{2}=|c|^{2}|\beta|^{2}\Rightarrow q=|c|^{2}\frac{|\beta|^{2}}{|\mu|^{2}}
	\end{equation}
	where $|c|^{2}\leq1$ beacuse of $K_{0}^{\dagger}K_{0}\leq\mathbbm{1}$. So
	\begin{equation}
		q=|c|^{2}\frac{|\beta|^{2}}{|\mu|^{2}}\leq\frac{|\beta|^{2}}{|\mu|^{2}}=\frac{1-|\iinner{f}{\psi}|^{2}}{1-|\iinner{f}{\phi}|^{2}}
	\end{equation}
	and the maximal value of $q$ is obtained when $c=1$.
	
	We can construct a free operation for QST using the Kraus operators that maximize the conversion probability.
	\begin{equation}\label{44}
		K_{0}=
		\begin{pmatrix}
			a & 0\\ 0 & 1
		\end{pmatrix},
		K_{1}=\sqrt{\mathbbm{1}-K_{0}^\dag K_{0}}
	\end{equation}
	where $a$ is defined as:
	\begin{equation}
		\begin{aligned}
			a=\frac{|\beta||\mu|}{|\alpha||\nu|}e^{i[(\mu-\nu)-(\alpha-\beta)]}.
		\end{aligned}
	\end{equation}
	It is easy to verify that this is a free operation
	\begin{equation}
		K_{0}\ket{f}=a\ket{f},K_{1}\ket{f}=\sqrt{1-|a|^{2}}\ket{f}
	\end{equation}
	meanwhile $|a|=\frac{|\beta||\mu|}{|\alpha||\nu|}\leq1$, $p\leq1$.
	
	Therefore, the Kraus operator $K_{0}$ transforms the state $\ket{\psi}$ into $\ket{\phi}$ with probability  $p=\frac{1-|\iinner{f}{\psi}|^{2}}{1-|\iinner{f}{\phi}|^{2}}$ as follows:
\begin{widetext}
	\begin{equation}
		\begin{aligned}
			K_{0}\ket{\psi}\bra{\psi}K_{0}^{\dagger} &=
			\begin{pmatrix}
				|a|^{2}|\alpha|^{2} & a\alpha\beta^{\ast}\\ a^{\ast}\alpha^{\ast}\beta & |\beta|^{2}
			\end{pmatrix}\\ &=
			\begin{pmatrix}
				\frac{|\beta|^{2}|\mu|^{2}}{|\nu|^{2}|\alpha|^{2}|}|\alpha|^{2} & \frac{|\beta||\mu|}{|\alpha||\nu|}e^{i[(\mu-\nu)-(\alpha-\beta)]}|\alpha||\beta|e^{i(\alpha-\beta)} \\ \frac{|\beta||\mu|}{|\alpha||\nu|}e^{i[(\alpha-\beta)-(\mu-\nu)]}|\alpha||\beta|e^{i(\beta-\alpha)} & |\beta|^{2}
			\end{pmatrix}\\ &=
			\begin{pmatrix}
				\frac{|\beta|^{2}}{|\nu|^{2}}|\mu|^{2} & \frac{|\beta|^{2}}{|\nu|^{2}}|\mu||\nu|e^{i(\mu-\nu)} \\ \frac{|\beta|^{2}}{|\nu|^{2}}|\mu||\nu|e^{i(\nu-\mu)} & \frac{|\beta|^{2}}{|\nu|^{2}}|\nu|^{2}
			\end{pmatrix}\\ &=\frac{|\beta|^{2}}{|\nu|^{2}}|
			\begin{pmatrix}
				|\mu|^{2} & |\mu||\nu|e^{i(\mu-\nu)} \\ |\mu||\nu|e^{i(\nu-\mu)} & |\nu|^{2}
			\end{pmatrix}\\ &=p\ket{\phi}\bra{\phi}.
		\end{aligned}
	\end{equation}
\end{widetext}
\end{proof}

Theorem 5 enables the determination of the optimal transition probability between pure states in QST resource theory and establishes its relation to the geometric measure.

\subsection{Stochastic transformations for mixed states}
%
As a generalization of Theorem 5, we now present the maximal probability for converting a mixed state $\rho$ into a pure state $\ket{\psi}$ via QST free operations.

To avoid ambiguity, we explicitly specify the initial state $\rho$ and the target state $\ket{\psi}$ as follows:
\begin{equation}
	\begin{aligned}
		\rho=
		\begin{pmatrix}
			\rho_{11} & \rho_{12} \\ \rho_{12}^{\ast} & 1-\rho_{11}
		\end{pmatrix}=
		\begin{pmatrix}
			\rho_{11} & |\rho_{12}|e^{i\theta} \\ |\rho_{12}|e^{-i\theta} & 1-\rho_{11}
		\end{pmatrix},
	\end{aligned}
\end{equation}
\begin{equation}
	\begin{aligned}
		\ket{\psi}\bra{\psi} &=(\alpha\ket{f}+\beta\ket{f^{\perp}})(\alpha^{\ast}\ket{f}+\beta^{\ast}\ket{f^{\perp}})\\ &=
		\begin{pmatrix}
			|\alpha|^{2} & |\alpha||\beta|e^{i(\alpha-\beta)} \\ |\alpha||\beta|e^{i(\beta-\alpha)} & |\beta|^{2}
		\end{pmatrix}.
	\end{aligned}
\end{equation}

\begin{thrm}
	The maximum probability of converting mixed qubit state $\rho$ to pure qubit state $\ket{\psi}$ via free operations of QST is given by
	\begin{equation}
		\begin{aligned}			P(\rho\to\ket{\psi})\leq\min\{\frac{\Upsilon_{g}(\rho)}{\Upsilon_{g}(\ket{\psi})},1\}=\min\{\frac{1-\bra{f}\rho\ket{f}}{1-|\iinner{f}{\psi}|^{2}},1\}.\quad \ \
		\end{aligned}
	\end{equation}
\end{thrm}
\begin{proof}
	Similar to the proof of Theorem 5, We also consider the case $\rho$ has weaker texture than $\ket{\psi}$, that is
	\begin{equation}
		\bra{f}\rho\ket{f}\geq|\iinner{f}{\psi}|^{2}\ ,\  |\rho_{11}|\geq|\alpha|^{2}
	\end{equation}
	and $\frac{1-\bra{f}\rho\ket{f}}{1-|\iinner{f}{\psi}|^{2}}\leq1$.

As demonstrated in Theorem 5, we can assume
\begin{equation}
		\begin{aligned}
			K_{0}=\begin{pmatrix}
				a & b \\ 0 & c
			\end{pmatrix}\ ,\ K_{1}=\sqrt{\mathbbm{1}-K_{0}^{\dagger}K_{0}}
		\end{aligned}
	\end{equation}
We also require that
\begin{equation}
\begin{aligned}
K_{0}\rho K_{0}^{\dagger}=m\ket{\psi}\bra{\psi}+n\ket{f}\bra{f}
\end{aligned}
\end{equation}
This is because $\ket{f}\bra{f}$ is a non-texture state, and we are not concerned with the probability of its occurrence. Therefore, our objective is to obtain
	\begin{widetext}
		\begin{equation}
			\begin{aligned}
				K_{0}\rho K_{0}^{\dagger} &=
				\begin{pmatrix}
					a & b \\ 0 & c
				\end{pmatrix}
				\begin{pmatrix}
					\rho_{11} & |\rho_{12}|e^{i\theta} \\ |\rho_{12}|e^{-i\theta} & 1-\rho_{11}
				\end{pmatrix}
				\begin{pmatrix}
					a^{\ast} & 0 \\ b^{\ast} & c^{\ast}
				\end{pmatrix}\\ &=
				\begin{pmatrix}
					|a|^{2}\rho_{11}+a^{\ast}b\rho_{12}^{\ast}+ab^{\ast}\rho_{12}+|b|^{2}(1-\rho_{11}) & ac^{\ast}\rho_{12}+bc^{\ast}(1-\rho_{11}) \\ a^{\ast}c\rho_{12}^{\ast}+b^{\ast}c(1-\rho_{11}) & |c|^{2}(1-\rho_{11})
				\end{pmatrix}\\ &=
				\begin{pmatrix}
					m|\alpha|^{2}+n & m\alpha\beta^{\ast} \\ m\alpha^{\ast}\beta & m|\beta|^{2}
				\end{pmatrix}\\ &=m\ket{\psi}\bra{\psi}+n\ket{f}\bra{f}.
			\end{aligned}
		\end{equation}
	\end{widetext}
	The same can be obtained
	\begin{equation}
		\begin{aligned}
			m=|c|^{2}\frac{1-\rho_{11}}{|\beta|^{2}}\leq\frac{1-\rho_{11}}{|\beta|^{2}}=\frac{1-\bra{f}\rho\ket{f}}{1-|\iinner{f}{\psi}|^{2}}
		\end{aligned}
	\end{equation}
	and the maximal value of $m$ is obtained when $c=1$, here we note that $b$ does not affect the convert probability.
	
\end{proof}

In general, the conversion probability does not attain its maximum value, though for particular quantum states, the maximum is achievable.

\begin{cor}
	For mixed qubit state $\rho$ and pure qubit state $\ket{\psi}$ meet the condition
	\begin{equation}\label{56}
		|\rho_{12}|\geq\frac{|\alpha|}{|\beta|}(1-\rho_{11})
	\end{equation}
	the maximum probability of converting $\rho$ to $\ket{\psi}$ via free operations of QST is
	\begin{equation}
		\begin{aligned}			P(\rho\to\ket{\psi})=\min\{\frac{\Upsilon_{g}(\rho)}{\Upsilon_{g}(\ket{\psi})},1\}=\min\{\frac{1-\bra{f}\rho\ket{f}}{1-|\iinner{f}{\psi}|^{2}},1\}.\quad \ \
		\end{aligned}
	\end{equation}
\end{cor}

\begin{proof}
	If the initial state $\rho$ and the target state $\ket{\psi}$ meet the condition $|\rho_{12}|\geq\frac{|\alpha|}{|\beta|}(1-\rho_{11})$, we can provide a specific stochastic convert operator
	\begin{equation}\label{58}
		K_{0}=
		\begin{pmatrix}
			a & 0\\ 0 & 1
		\end{pmatrix}\ ,\
		K_{1}=\sqrt{\mathbbm{1}-K_{0}^\dag K_{0}}
	\end{equation}
	where $a$ is defined as:
	\begin{equation}
		a=\frac{(1-\rho_{11})|\alpha|}{|\rho_{12}||\beta|}e^{i(\alpha-\beta-\theta)}
	\end{equation}
	and $|a|=\frac{(1-\rho_{11})|\alpha|}{|\rho_{12}||\beta|}\leq1$ according to the condition we set.
	
	Then we can show that state $\rho$ transform to state $\ket{\psi}$ with the probability $p=\frac{1-\bra{f}\rho\ket{f}}{1-|\iinner{f}{\psi}|^{2}}$ after the action of $K_{0}$.
	\begin{widetext}
		\begin{equation}
			\begin{aligned}
				K_{0}\rho K_{0}^{\dagger} &=
				\begin{pmatrix}
					|a|^{2}\rho_{11} & a\rho_{12}\\ a^{\ast}\rho_{12}^{\ast} & 1-\rho_{11}
				\end{pmatrix}\\ &=
				\begin{pmatrix}
					\frac{(1-\rho_{11})^{2}|\alpha|^{2}}{|\rho_{12}|^{2}|\beta|^{2}}\rho_{11} & \frac{(1-\rho_{11})|\alpha|}{|\rho_{12}||\beta|}e^{i(\alpha-\beta-\theta)}|\rho_{12}|e^{i\theta} \\ \frac{(1-\rho_{11}|\alpha|)}{|\rho_{12}||\beta|}e^{i(-\alpha+\beta+\theta)}|\rho_{12}|e^{-i\theta} & (1-\rho_{11})
				\end{pmatrix}\\ &=
				\begin{pmatrix}
					\frac{1-\rho_{11}}{|\beta|^{2}}|\alpha|^{2}+\frac{(1-\rho_{11})|\alpha|^{2}[(1-\rho_{11})\rho_{11}-|\rho_{12}|^{2}]}{|\rho_{12}|^{2}|\beta|^{2}} & \frac{1-\rho_{11}}{|\beta|^{2}}|\alpha||\beta|e^{i(\alpha-\beta)} \\ \frac{1-\rho_{11}}{|\beta|^{2}}|\alpha||\beta|e^{i(\beta-\alpha)} & \frac{1-\rho_{11}}{|\beta|^{2}}|\beta|^{2}
				\end{pmatrix}\\ &=\frac{1-\rho_{11}}{|\beta|^{2}}
				\begin{pmatrix}
					|\alpha|^{2} & \alpha\beta^{\ast} \\ \alpha^{\ast}\beta & |\beta|^{2}
				\end{pmatrix}+\frac{(1-\rho_{11})|\alpha|^{2}[(1-\rho_{11})\rho_{11}-|\rho_{12}|^{2}]}{|\rho_{12}|^{2}|\beta|^{2}}\begin{pmatrix}
					1 & 0 \\ 0 & 0
				\end{pmatrix} \\ &=p\ket{\psi}\bra{\psi}+q\ket{f}\bra{f}.
			\end{aligned}
		\end{equation}
	\end{widetext}
\end{proof}

Therefore, for a chosen target state $|\psi\rangle$, if the initial state is pure, there exists a free operator (\ref{44}) that maximizes the conversion probability. If the initial state is mixed, one must first determine whether condition (\ref{56}) is satisfied. If it is, the maximum conversion probability can be achieved using the operator (\ref{58}) that we provide.

\section{Relationships between texture and other quantum properties}

The relationships between quantum resources constitute an important research direction in quantum resource theory \cite{ref37,ref38,ref39,ref40,ref41}. A connection between QST and coherence has been established \cite{ref1}, as the non-texture state coincides with the maximally coherent state. Furthermore, both resources, along with imaginarity and predictability, are basis-dependent, requiring a fixed reference frame for their definition. This commonality suggests a fundamental relationship between QST, coherence, imaginarity, and predictability.

Resource measures defined by the minimum distance between quantum states and free states can be intuitively visualized on the Bloch sphere \cite{ref40,ref41,ref21}. Therefore, we primarily analyze relationships based on the $l_{1}$-norm and $l_{2}$-norm, and attempt to illustrate them geometrically on the Bloch sphere.

\subsection{Relationship under $l_{1}$-norm measure}

Quantum state $\rho$ have a pauil decomposition
\begin{equation}
	\rho=\frac{1}{2}(\mathbbm{1}+x\sigma_{x}+y\sigma_{y}+z\sigma_{z})
\end{equation}
when $d=2$. At this point, the quantum state can be visually represented as a point on the Bloch sphere \cite{ref27}
\begin{equation}
	\rho=\vec{r}\cdot\vec{\sigma}
\end{equation}
where $\vec{r}=(x,y,z)$, and $\vec{\sigma}=(\sigma_{x},\sigma_{y},\sigma_{z})$. And the $l_{1}$-norm measure of QST can be expressed as \cite{ref12,ref27}:
\begin{equation}
	\begin{aligned}
		\Upsilon_{l_{1}}(\rho)=D_{l_{1}}(\rho,f) &=\tr{|\rho-f|}\\ &=\tr{|(\vec{r}-\vec{s})\vec{\sigma}|}\\ &=|\vec{r}-\vec{s}|\\ &=\sqrt{(x-1)^{2}+y^{2}+z^{2}}.
	\end{aligned}
\end{equation}

Meanwhile, refer to FIG. \ref{figure1}, the coherence measure $C_{l_{1}}$; imaginarity meaure $I_{l_{1}}$; predictability measure $P_{l_{1}}$ defined by $l_{1}$-norm can be expressed as follows \cite{ref16,ref21,ref41}:
\begin{equation}
	\begin{aligned}
		&C_{l_{1}}(\rho)=D_{l_{1}}(\rho,\Delta(\rho))=\sqrt{x^{2}+y^{2}};\\
		&I_{l_{1}}(\rho)=D_{l_{1}}(\rho,\tilde{\rho})=\sqrt{y^{2}};\\
		&P_{l_{1}}(\rho)=D_{l_{1}}(\rho,\overline{\rho})=\sqrt{z^{2}}
	\end{aligned}
\end{equation}
where $\Delta(\rho)$, $\tilde{\rho}$ and $\overline{\rho}$ are the closest incoherence state, real state and non-predictable state to $\rho$ respectively. Apparently
\begin{equation}\label{66}
	\begin{aligned}
		\Upsilon_{l_{1}}(\rho)
		&=\sqrt{(x-1)^{2}+y^{2}+z^{2}}\\
		&\leq\sqrt{(x-1)^{2}}+\sqrt{y^{2}}+\sqrt{z^{2}}\\
		&=1-x+\sqrt{y^{2}}+\sqrt{z^{2}}\\
		&=1-x+I_{l_{1}}(\rho)+P_{l_{1}}(\rho)
	\end{aligned}
\end{equation}
beacause $x\in[-1,1]$, so $x-1$ must be negative. Here we get a relationship between QST, imaginarity and predictability.

\begin{figure}[h]
	\centering
	\includegraphics[width=0.5\textwidth]{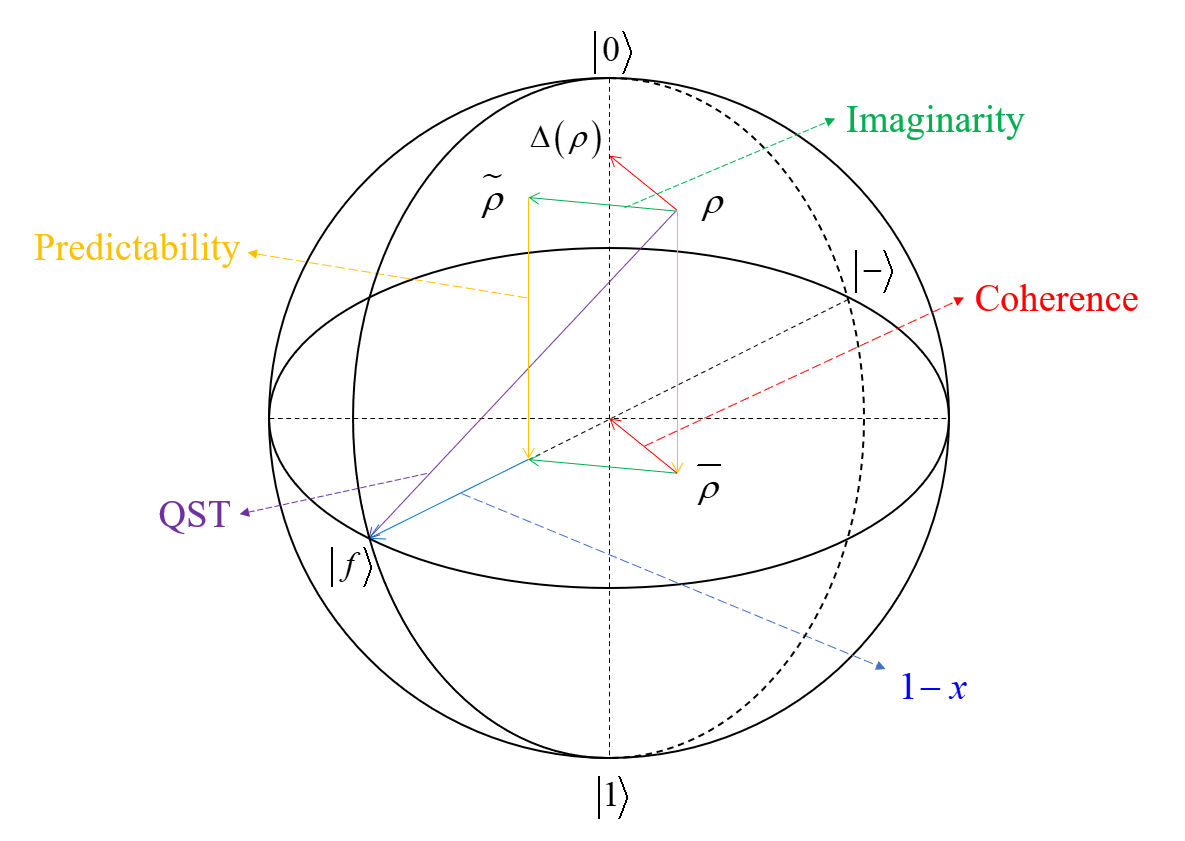}
	\caption{The purple vector (QST) can be decomposed into the vector sum of the yellow vector (Predictability), the green vector (Imaginarity) and the blue vector ($1-x$), serving as the intuitive representation on the Bloch sphere.}\label{figure1}
\end{figure}

Inequality (\ref{66}) reflects only a numerical relationship, as the parameters are treated as scalars. However, quantities such as QST, imaginarity, predictability, and even $1-x$ possess specific directions in the Bloch sphere representation, corresponding to different quantum state properties as illustrated in Fig. \ref{figure1}. By intuitively considering these parameters as vectors, we obtain
\begin{equation}
	\begin{aligned}
		\overrightarrow{\Upsilon_{l_{1}}(\rho)}=\overrightarrow{I_{l_{1}}(\rho)}+\overrightarrow{P_{l_{1}}(\rho)}+\overrightarrow{1-x}.
	\end{aligned}
\end{equation}
Naturally, Ineq. (\ref{66}) can also be obtained because of
\begin{equation}
	\vec{a}+\vec{b}=\vec{c}\Rightarrow a+b\geq c.
\end{equation}
When $\rho$ is on the X-axis, that is, $\rho$ is a linear combination of $\ket{+}\bra{+}$ and $\ket{-}\bra{-}$, the inequality becomes an equality
\begin{equation}
	\begin{aligned}
		\Upsilon_{l_{1}}(\rho)=\sqrt{(x-1)^{2}}=1-x
	\end{aligned}
\end{equation}
and corresponding to the intuitive representation on Bloch sphere $\overrightarrow{\Upsilon_{l_{1}}(\rho)}=\overrightarrow{1-x}$.

Seeking a more precise characterization, we now derive the exact relationship. Revisiting $\Upsilon_{l_{1}}$, when $x\geq0$ (i.e., for quantum states $\rho$ located in the front hemisphere of the Bloch sphere), we find
\begin{equation}
	\begin{aligned}
		\Upsilon_{l_{1}}(\rho)
		&\leq1-x+\sqrt{y^{2}}+\sqrt{z^{2}}\\
		&=1-\sqrt{x^{2}}+\sqrt{y^{2}}+\sqrt{z^{2}}\\
		&=1-(\sqrt{x^{2}}+\sqrt{y^{2}})+2\sqrt{y^{2}}+\sqrt{z^{2}}\\
		&=1-C_{l_{1}}(\rho)+2I_{l_{1}}(\rho)+P_{l_{1}}(\rho).
	\end{aligned}
\end{equation}
Here, QST exhibits a positive correlation with imaginarity and predictability, but a negative correlation with coherence. Furthermore, as evident from FIG. \ref{figure1}, the magnitude of QST (purple line) depends on the magnitudes of predictability (yellow line) and imaginarity (green line), while exhibiting a dual relationship with coherence (red line). In particular, when $\rho$ lies on the positive X-axis, $\Upsilon_{l_{1}}$ reaches its maximum value. At this point,
\begin{equation}
	\Upsilon_{l_{1}}=1-x=1-C_{l_{1}}.
\end{equation}

\begin{figure}[h]
	\centering
	\includegraphics[width=0.5\textwidth]{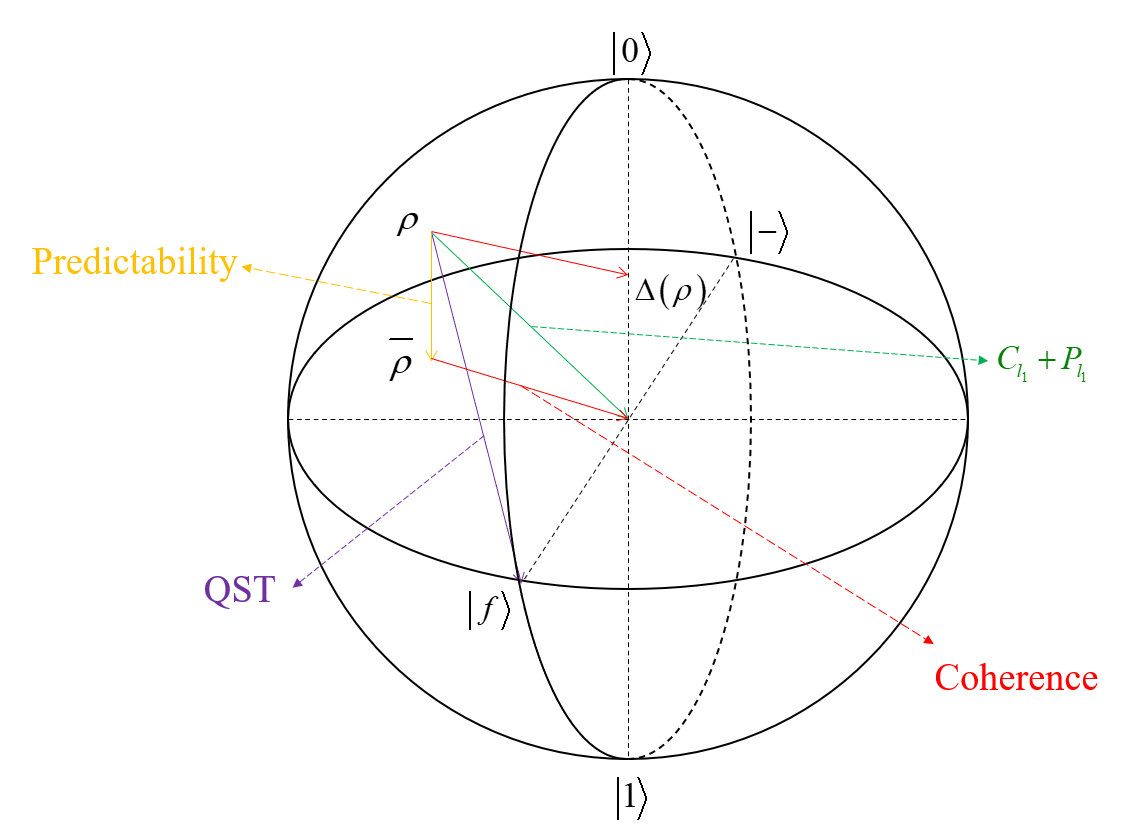}
	\caption{The vector sum of the red vector (Coherence) and the yellow vector (Predictability) constitutes the green vector ($C_{l_{1}}+P_{l_{1}}$), while the vector sum of the green vector ($C_{l_{1}}+P_{l_{1}}$) and the positive x-axis constitutes the purple vector (QST), serving as the intuitive representation on the Bloch sphere.}\label{figure2}
\end{figure}

When $x<0$, that is, when the quantum state $\rho$ is in the rear hemisphere, we get
\begin{equation}
	\begin{aligned}
		\Upsilon_{l_{1}}(\rho)
		&\leq1-x+\sqrt{y^{2}}+\sqrt{z^{2}}\\
		&=1+\sqrt{x^{2}}+\sqrt{y^{2}}+\sqrt{z^{2}}\\
		&=1+C_{l_{1}}(\rho)+P_{l_{1}}(\rho).
	\end{aligned}
\end{equation}
At this point, QST is determined solely by coherence and predictability, since coherence already incorporates imaginarity. This conclusion can also be derived from the vector relationship depicted in FIG. \ref{figure2}. Intuitively, the magnitude of QST (purple line) depends on the magnitudes of coherence (red line) and predictability (yellow line). Similarly, when $\rho$ is located on the negative X-axis, $\Upsilon_{l_{1}}$ reaches its maximum value. Under these conditions,
\begin{equation}
	\Upsilon_{l_{1}}=1+C_{l_{1}}.
\end{equation}
The preceding analysis treats all parameters as scalars and thus reflects only numerical relationships.

We now proceed to perform different decompositions of $\Upsilon_{l_{1}}$:
\begin{equation}
	\begin{aligned}
		\Upsilon_{l_{1}}^{2}(\rho)
		&=(x-1)^{2}+y^{2}+z^{2}\\
		&=x^{2}+y^{2}+z^{2}+(1-2x)\\
		&=C_{l_{1}}^{2}(\rho)+P_{l_{1}}^{2}(\rho)+(1-2x).
	\end{aligned}
\end{equation}
By control the range of $x$, we can obtain the upper or lower limits of the duality relationship.
\begin{equation}
	\begin{cases}
		\ C_{l_{1}}^{2}+P_{l_{1}}^{2}\leq\Upsilon_{l_{1}}^{2}\ ,&\ x\in[-1,\frac{1}{2})\\
		\ C_{l_{1}}^{2}+P_{l_{1}}^{2}=\Upsilon_{l_{1}}^{2}\ ,&\ x=\frac{1}{2}\\
		\ C_{l_{1}}^{2}+P_{l_{1}}^{2}\geq\Upsilon_{l_{1}}^{2}\ ,&\ x\in(\frac{1}{2},1]\\
	\end{cases}
\end{equation}

Due to the constraints imposed by the $l_{1}$-norm formulation, our analysis of relationships between quantum resources is limited to $d=2$. For $d>2$, the decomposition of $\Upsilon_{l_{1}}$ becomes considerably more complex, as each $x_{i}$ requires individual evaluation. We therefore seek alternative approaches to establish relationships in higher dimensions.
\subsection{Relationship under $l_{2}$-norm measure}

Now we try to establish relationship under $l_{2}$-norm measure. We defined a QST measure
\begin{equation}
	\Upsilon_{l_{2}}(\rho)=||\rho-f||_{l_{2}}^2=\sum_{i,j}|\rho_{ij}-f_{ij}|^{2}=\sum_{i,j}|\rho_{ij}-\frac{1}{d}|^{2}.\quad
\end{equation}
While this is intuitively a good measure of QST, but it has some drawbacks. According to the properties of $L_{2}$-norm, non-negativity and convexity are obvious, however
\begin{equation}
	\|\Lambda(\rho)-f\|_{l_{2}}=\|\Lambda(\rho)-\Lambda(f)\|_{l_{2}}=\|\Lambda(\rho-f)\|_{l_{2}}=\|\rho-f\|_{l_{2}}\quad
\end{equation}
because of $\|A\|_{l_{2}}=\|\Phi(A)\|_{l_{2}}$, $\Lambda$ is free operation among them. Whatever, this is an acceptable QST measure.

We will now decompose $\Upsilon_{l_{2}}$, get that
%
%
\begin{equation}
	\begin{aligned}
		\Upsilon_{l_{2}}(\rho) &=\sum_{i,j=1}^{d}|\rho_{ij}-\frac{1}{d}|^{2}\\ &=\sum_{i,j=1}^{d}(\rho_{ij}-\frac{1}{d})(\rho_{ij}^{\ast}-\frac{1}{d})\\ &=\sum_{i,j=1}^{d}(|\rho_{ij}|^{2}+\frac{1}{d^{2}}-\frac{1}{d}\rho_{ij}-\frac{1}{d}\rho_{ij}^{\ast})\\ &=1+\sum_{i,j=1}^{d}|\rho_{ij}|^{2}-\frac{2}{d}\sum_{i,j=1}^{d}\Re{(\rho_{ij})}\\
		&=\sum_{i\neq j}|\bra{i}\rho\ket{j}|^{2}+\sum_{i=1}^{d}|\bra{i}\rho\ket{i}|^{2}+(1-2\bra{f}\rho\ket{f})\\
		&=C_{l_{2}}(\rho)+P_{l_{2}}(\rho)+(1-2\bra{f}\rho\ket{f}).\quad
	\end{aligned}
\end{equation}

Among them, $C_{l_{2}}(\rho)=\sum_{i\neq j}|\bra{i}\rho\ket{j}|^{2}$ and $P_{l_{2}}(\rho)=\sum_{i}|\bra{i}\rho\ket{i}|^{2}$ are coherence and predictability defined by $l_{2}$-norm respectively \cite{ref42}, and
\begin{equation}\label{81}
	\bra{f}\rho\ket{f}=\frac{1}{d}\sum_{i,j=1}^{d}\rho_{ij}
\end{equation}
is inversely correlated with QST \cite{ref1}.

Here we notice that
\begin{equation}
	1-2\bra{f}\rho\ket{f}=2\Upsilon_{g}(\rho)-1
\end{equation}
and $\Upsilon_{g}$ is proposed in Example 4 above, then we can get
\begin{equation}
	C_{l_{2}}(\rho)+P_{l_{2}}(\rho)+2\Upsilon_{g}(\rho)=\Upsilon_{l_{2}}(\rho)+1.
\end{equation}
And to be able to get a new wave-particle relationship
\begin{equation}
	C_{l_{2}}(\rho)+P_{l_{2}}(\rho)=\Upsilon_{l_{2}}(\rho)-2\Upsilon_{g}(\rho)+1.
\end{equation}

To summarize, the $l_{2}$-norm enables investigation of numerical relationships across all dimensions. Notably, while QST consistently connects with coherence and predictability, imaginarity vanishes from these relationships---a phenomenon likely arising from the mathematical structure of $l_{2}$-norm.

\subsection{Relationship under skew information}

Quantum skew information plays a significant role in quantum information theory and has been used to quantify coherence \cite{ref43}. We therefore propose that it can also be employed to establish relationships between quantum resources.

Quantum skew information is defined as
\begin{equation}
	I(\rho,\ket{k}\bra{k})=-\frac{1}{2}\tr{\rho,\ket{k}\bra{k}}^{2}.
\end{equation}
And we are able to construct the skew information of $\ket{f}\bra{f}$
\begin{equation}
	I(\rho,\ket{f}\bra{f})=-\frac{1}{2}\tr{\rho,\ket{f}\bra{f}}^{2}.
\end{equation}
Decompose $I(\rho,\ket{f}\bra{f})$, we can get
\begin{widetext}
	\begin{equation}
		\begin{aligned}
			I(\rho,\ket{f}\bra{f}) &=-\frac{1}{2}\tr{[\sqrt{\rho},\ket{f}\bra{f}]}^{2}\\ &=-\frac{1}{2}\tr{[\sqrt{\rho},\ket{f}\bra{f}][\sqrt{\rho},\ket{f}\bra{f}]^{\dagger}}\\ &=-\frac{1}{2}\trace[(\sqrt{\rho}\ket{f}\bra{f}-\ket{f}\bra{f}\sqrt{\rho})(\ket{f}\bra{f}\sqrt{\rho}-\sqrt{\rho}\ket{f}\bra{f})]\\ &=-\frac{1}{2}\tr{\sqrt{\rho}\ket{f}\bra{f}\sqrt{\rho}+\ket{f}\bra{f}\rho\ket{f}\bra{f}-\ket{f}\bra{f}\sqrt{\rho}\ket{f}\bra{f}\sqrt{\rho}-\sqrt{\rho}\ket{f}\bra{f}\sqrt{\rho}\ket{f}\bra{f}}\\ &=-\frac{1}{2}[\tr{\rho\ket{f}\bra{f}}+\tr{\ket{f}\bra{f}\rho\ket{f}\bra{f}}-\tr{\ket{f}\bra{f}\sqrt{\rho}\ket{f}\bra{f}\sqrt{\rho}}-\tr{\sqrt{\rho}\ket{f}\bra{f}\sqrt{\rho}\ket{f}\bra{f}}]\\
			&=-\frac{1}{2}[\bra{f}\rho\ket{f}+\bra{f}\rho\ket{f}-(\bra{f}\sqrt{\rho}\ket{f})^{2}-(\bra{f}\sqrt{\rho}\ket{f})^{2}]\\  &=(\bra{f}\sqrt{\rho}\ket{f})^{2}-\bra{f}\rho\ket{f}\\ &=(A(\rho,f))^{2}-\frac{1}{d}-\frac{1}{d}\sum_{i\neq j}\rho_{ij},
		\end{aligned}
	\end{equation}
\end{widetext}
we used Eq. \eqref{81} here. Then we can get
\begin{equation}
	\begin{aligned}
		A(\rho,f)^{2}-\bra{f}\rho\ket{f}=I(\rho,\ket{f}\bra{f}).
	\end{aligned}
\end{equation}
The affinity
\begin{equation}
A(\rho, f) = \bra{f}\sqrt{\rho}\ket{f} = \sqrt{I(\rho, \ket{f}\bra{f}) + \bra{f}\rho\ket{f}}
\end{equation}
is also directly measurable, since both $I(\rho, \ket{f}\bra{f})$ and $\bra{f}\rho\ket{f}$ can be measured directly. This implies that the QST measure $\Upsilon_{H}$, defined via the Hellinger distance, can also be obtained through direct measurement. Similarly, an inequality can be derived:
\begin{equation}
	A(\rho,f)^{2}-I(\rho,\ket{f}\bra{f})-\frac{C_{l_{1}}(\rho)}{d}\leq\frac{1}{d}
\end{equation}
base on
\begin{equation}
	\begin{aligned}
			\sum_{i\neq j}\rho_{ij}\leq\sum_{i\neq j}|\rho_{ij}| =C_{l_{1}}(\rho).
		\end{aligned}
\end{equation}

We establish relationships among directly measurable quantities---$\bra{f}\rho\ket{f}$, affinity, and skew information---and incorporate coherence to derive an inequality. This inequality, based on the affinity and skew information relative to the texture, provides a measurable lower bound for the $l_1$-norm of coherence, $C_{l_{1}}$.
\section{Conclusion}

We have supplemented the framework with two new definitions of QST measures and have theoretically demonstrated that they satisfy the three fundamental conditions for a quantum-state texture measure, as outlined in Ref. \cite{ref1}. We also propose a general functional form for constructing QST measures via the convex roof method, which may play a key role in addressing certain measurement problems.

Furthermore, we discuss the maximal probability of state transformation via free operations, deriving the optimal transformation probability between pure states as well as the maximum probability for transforming mixed states into pure states.

Finally, we investigate the relationship between QST and other prominent quantum resources, establishing a series of equations and inequalities that quantify the connections among QST, coherence, imaginarity, and predictability from multiple perspectives. This work thereby provides valuable insights into the physical interrelations of quantum resources. We believe that the results presented in this paper contribute to the enrichment and advancement of QST resource theory and lay a solid theoretical foundation for its future development.

\begin{acknowledgments}
This work is supported by the Fundamental Research Funds for the Central Universities (Grants No. 3072025YC2404).
\end{acknowledgments}


%

\end{document}